\documentclass[11pt]{article}
\usepackage[a4paper,top=3cm,bottom=3.5cm,left=2.5cm,right=2.5cm,marginparwidth=1.75cm]{geometry}
\usepackage[T1]{fontenc}
\usepackage{graphicx}
\usepackage[title]{appendix}
\usepackage{tikz}
\usepackage{pgfplots}
\usepackage{cutwin}

\usepackage{amsmath,yhmath}
\usepackage{amsthm}
\usepackage{amssymb}
\usepackage{hyperref}
\usepackage{cleveref}
\usepackage{epstopdf}
\usepackage{comment}
\usepackage{todonotes}
\usepackage{color}
\usepackage{float}
\usepackage[sort&compress,numbers]{natbib}
\usetikzlibrary{shapes.geometric}
\usetikzlibrary{arrows,automata,quotes}
\usepackage{iftex}
\usepackage{xfrac} 
\usepackage{amsmath,yhmath}
\usepackage{amsthm}
\usepackage{amssymb}
\usepackage{mathtools}
\usepackage{nccmath}
\usepackage{hyperref}
\usepackage{caption, subcaption}
\usepackage[title]{appendix}
\usepackage{cleveref}
\usepackage{epstopdf}
\usepackage{comment}
\usepackage{color}
\usepackage{float}
\usepackage{cutwin}
\usepackage[sort&compress,numbers]{natbib}
\usepackage{makecell}
\usepackage{array}   
\usepackage{multirow}
\usepackage{graphicx}
\usepackage{caption}
\usepackage{subcaption}
\newtheorem{thm}{Theorem}[section]
\newtheorem{cor}[thm]{Corollary}
\newtheorem{lemma}[thm]{Lemma}
\newtheorem{prop}[thm]{Proposition}
\newtheorem{ass}[thm]{Assumption}
\theoremstyle{definition}

\newtheorem{rem}[thm]{Remark}

\renewcommand{\Re}{\mathrm{Re}}
\newcommand{\inc}{\mathrm{in}}
\newcommand{\s}{\mathrm{sc}}
\newcommand{\p}{\partial}
\newcommand{\ds}{\displaystyle}

\newcommand{\Z}{\mathbb{Z}}

\numberwithin{equation}{section}
\numberwithin{figure}{section}
\title{Effective Medium Theory for Time-modulated Subwavelength Resonators}
\author{Habib Ammari\thanks{\footnotesize Department of Mathematics, ETH Z\"urich, R\"amistrasse 101, CH-8092 Z\"urich, Switzerland (habib.ammari@math.ethz.ch, liora.rueff@sam.math.ethz.ch).}\and Jinghao Cao\thanks{\footnotesize Computing and Mathematical Sciences, California Institute of Technology, Pasadena, CA 91125, USA (jinghao.cao@caltech.edu).} \and Erik Orvehed Hiltunen\thanks{\footnotesize Department of Mathematics, University of Oslo, Moltke Moes vei 35, 0851 Oslo, Norway (erikhilt@math.uio.no).} \and Liora Rueff\footnotemark[1] }
\date{}

\begin{document}
\maketitle
\begin{abstract}
    This paper provides a general framework for deriving effective material properties of one-dimensional, time-modulated systems of subwavelength resonators. It applies to subwavelength resonator systems with a general form of time-dependent parameters. We show that the resonators can be accurately described by a point-scattering formulation when the width of the resonators is small. In contrast to the static setting, where this point interaction approximation yields a Lippmann-Schwinger equation for the effective material properties, the mode coupling in the time-modulated case instead yields an infinite linear system of Lippmann-Schwinger-type equations. The effective equations can equivalently be written as a system of differential equations. Moreover, we introduce a numerical scheme to approximately solve the system of coupled equations and illustrate the validity of the effective equation.
\end{abstract}
\noindent{\textbf{Mathematics Subject Classification (MSC2000):} 35J05, 35C20, 35P20, 74J20}
\vspace{0.2cm}

\section{Introduction}

This paper is devoted to deriving the homogenised equation of one-dimensional, time-dependent metamaterials. The theory of homogenisation is central to a plethora of physical and engineering applications \cite{milton}. Homogenisation provides a way to transition from the detailed, small-scale description to an effective, large-scale description, simplifying analysis and computation while preserving essential features, and is a cornerstone method in the theory of metamaterials \cite{choy2015effective}. With significant attention devoted to metamaterials that are periodically driven in time, resulting in time-modulated material parameters, a natural problem is to develop homogenisation theories for time-dependent problems. We specifically choose to study one-dimensional materials since they allow a more detailed exploration of the effect of time-modulations. Several results have been established in the one-dimensional, time-dependent setting, notably the capacitance matrix approximation \cite{jinghao_liora} and the rigorous definition of the scattered wave field for a given incident wave field \cite{ammari2024scattering}.

This paper provides a general framework for the effective medium theory of time-modulated subwavelength metamaterials. The propagation of waves through these media is governed by the same equations as posed in \cite{ammari2024scattering}. Specifically, we assume that one of the material parameters $\kappa$ varies in time inside the resonators, while $\rho$ is a static material parameter (see Section \ref{sec:problem_setting} for a definition of the setting). We opt to only consider materials with time-dependent $\kappa$ and static $\rho$ since previous work has shown that modulating $\rho$ in time does not affect wave propagation at leading order \cite{jinghao_liora}. Our theory applies to subwavelength resonator systems where the modulation frequency is of the same order as the subwavelength quasifrequencies and the operating frequency (\textit{i.e.}, the frequency of the incident wave) is in the low-frequency regime. It also generalises to higher space dimensions, as outlined in Remark \ref{rem:higherdim}. 

In contrast to this paper, previous works have successfully established an effective medium theory for the low-frequency regime, in which the operating frequency is significantly smaller than the resonant frequency and the material parameters are static in time \cite{Caflisch_Miksis_Papanicolaou_Ting_1985,Caflisch_Miksis_Papanicolaou_Ting_1985_waveprop,Kargl2002}. Furthermore, it is worth emphasising that the time-modulations considered in this paper are very different from the travelling wave form modulations discussed in \cite{liu,rizza,pendry1,touboul2024high} and in the references therein, where the mathematical analysis is simplified by passing to a moving coordinate frame. Opposed to the time-modulation in travelling wave form, our setting leads to the coupling of different frequency harmonics and is characterised by a system of coupled differential equations. By exploiting the subwavelength resonance of the building blocks, such systems may exhibit subwavelength resonant quasifrequencies \cite{ammari2024scattering}, spatiotemporally localised modes \cite{liora2024st_localisation}, $k$-gaps and unidirectional band gaps \cite{jinghao1}. Moreover, their mathematical treatment is more involved. In the static setting, numerous papers have established an effective medium theory for subwavelength resonator systems based on the point interaction approximation \cite{pointapp1,pointapp2} resulting in the Lippmann-Schwinger equation \cite{Hai_Habib,florian23,florian_MonteCarlo,Ammari_Challa,Foldy_1945}. In contrast to the static setting, we are faced with a coupled problem in the time-modulated case, which yields an infinite linear system of Lippmann-Schwinger-type equations instead. 
Moreover, implementing the point interaction approximation under the assumption of time-dependent material parameters is highly non-trivial, as there are two subwavelength quasifrequencies for a single time-modulated resonator. We shall see that, under suitable assumptions on volume fraction, configuration and incident frequency, we can adapt the results valid in the static case \cite{Hai_Habib} to hold true in the time-modulated case.

In this paper, we follow an approach that involves the scattering matrix, which we shall define from first principles for both static and time-dependent metamaterials \cite{Lin_2022}. We show that the scattering matrix has a characteristic structure in the leading-order terms as the resonators become small, which furnishes the point interaction approximation. We recall that the idea of a point interaction approximation goes back
to Foldy’s paper \cite{Foldy_1945}. It is a natural tool to analyse a variety of interesting problems in the continuum limit. Taking the continuum limit allows us to derive a theory valid for infinitely many resonators, such as in \cite{Caflisch_Miksis_Papanicolaou_Ting_1985} and \cite{pointapp1}, since it acts as a tool to average over an inhomogeneous material.

Our point-scattering formulation is highly relevant in its own, as the scattering matrix describes the coupling of different frequency harmonics and quantifies the frequency conversion due to the time-modulated parameters. Moreover, by taking the continuum limit, we derive the homogenised governing ordinary differential equation, which models the effective macroscopic behaviour while averaging out the fine-scale variations. This allows the introduction of a time-dependent, one-dimensional effective medium theory, \textit{i.e.}, a theoretical framework to describe the macroscopic properties of heterogeneous time-dependent, one-dimensional metamaterials in terms of their microscopic structure. The assumption of time-modulated material parameters in the derivation of an effective medium theory marks a new milestone in the mathematical exploration of metamaterials, where classical wave-scattering results are generalised to the time-modulated setting. It is worth emphasising that while our purpose in \cite{jinghao_liora, ammari2024scattering} was the study of wave scattering from systems of time-modulated systems of subwavelength resonators, our aim in the present paper is to derive effective models for computing the scattered fields by large systems of time-modulated resonators in the limit where the number of resonators goes to infinity and their typical size goes to zero, keeping the volume fraction of the resonators constant.
Such non-classical models, which capture the average macroscopic behaviour of the large systems of resonators, can be used to simplify the analysis of time-modulated metamaterials and make the computations of their scattering properties more feasible and stable. Note also that because of the subwavelength resonant phenomena, the effective models both in the static case and in the time-modulated case depend on the operating frequency. Moreover, the effective models are not valid when the operating frequency coincides with the real part of one of the resonant frequencies \cite{Hai_Habib}. In such situations, the effect of each resonator on the total scattered wave is of order one, and hence the limit does not exist as the number of resonators goes to infinity. The main difference between the static case and the time-modulated case comes from the fact that time-modulation gives rise to a family of coupled harmonics. This leads to an effective model described by a system of coupled differential equations instead of a Helmholtz equation with an effective (frequency dependent) potential as in the static case.

Our paper is organised as follows. We start by providing an overview of the mathematical and physical setting of the problem considered herein in Section \ref{sec:problem_setting}. In Section \ref{sec:static_material} we focus on one-dimensional metamaterials with static properties and derive the homogenised equation by introducing a scattering matrix and exploiting its structure. Although our main focus is on the time-modulated case, the static metamaterials serve as a concise introduction to our methods and ideas. This method is used to obtain an effective medium theory in Section \ref{sec:timedep_material} for time-modulated metamaterials. To achieve this, we explicitly compute the scattering matrix as the resonators become smaller. Finally, we solve the effective equation numerically in Section \ref{sec:numerical_results} with a uniquely tailored numerical scheme. We summarise our results and conclusions in Section \ref{sec:conclusion}.

\section{Problem Setting}\label{sec:problem_setting}
We consider a finite medium $\mathcal{U}\subset\mathbb{R}$, which contains $N$ disjoint high-contrast subwavelength resonators $\left(D_i\right)_{i=1,\dots,N}$. Each resonator is defined as an open interval $D_i:=\left(x_i^-,x_i^+\right)$ of length $\ell_i:=x_i^+-x_i^-$ with a separation distance $\ell_{i(i+1)}:=x_{i+1}^--x_i^+$, for all $i=1,\dots,N$. We denote the centre of each resonator by $z_i=(x_i^-+x_i^+)/2$. In the remainder of this paper we shall assume that each resonator is of length $\ell$ and that they are evenly spaced. We denote the disjoint union of all resonators $D_i$ by $D$. We refer to Figure \ref{fig:setup} for an illustration of the geometrical setup of the material.
\begin{figure}[H]
		\centering
		\includegraphics[width=0.5\textwidth]{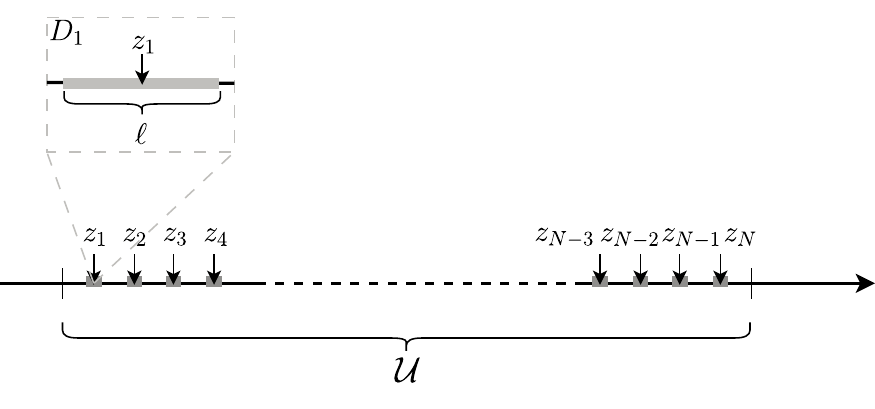}
		\caption{An illustration of the domain $\mathcal{U}$ with $N$ resonators.}
		\label{fig:setup}
	\end{figure}
We denote the material parameters inside the resonator $D_i$ by $\kappa_{\mathrm{r}}\kappa_i(t)$ and $\rho_{\mathrm{r}}$. We write:
\begin{align}
    \rho(x)=\begin{cases} 
    \rho_0, & x \notin {D}, \\
    \rho_{\mathrm{r}}, & x \in D_i,
    \end{cases} \qquad \kappa(x,t)=\begin{cases} 
    \kappa_0, & x \notin {D}, \\
    \kappa_{\mathrm{r}} \kappa_i(t), & x \in D_i.
    \end{cases}
\end{align}
For our numerical simulations, we shall consider
\begin{equation}
    \kappa_i(t):=\frac{1}{1+\varepsilon_{\kappa,i}\cos\left(\Omega t+\phi_{\kappa,i}\right)}, \label{eq:rho_kappa}
\end{equation}
for all $1\leq i\leq N$, where $\Omega$ is the frequency of the time-modulations, $\varepsilon_{\kappa,i} \in [0,1)$ are their amplitudes, and $\phi_{\kappa,i} \in [0,2\pi)$ are the phase shifts.\par 

 Identically as in \cite{ammari2024scattering}, we introduce the contrast parameter and the wave speeds
\begin{equation} \label{defdelta}
    \delta:=\frac{\rho_{\mathrm{{r}}}}{\rho_0},\quad v_0:=\sqrt{\frac{\kappa_0}{\rho_0}},\quad v_{\mathrm{r}}:=\sqrt{\frac{\kappa_\mathrm{r}}{\rho_{\mathrm{r}}}},
\end{equation}
respectively. 

We consider an incident wave field $u^{\inc}(x,t)$ with a (real) frequency $\omega$ and such that 
\begin{align}\label{eq:PDE_uin}
    \frac{\partial^2}{\partial t^2} u^{\inc}(x,t)-v_0^2\Delta u^{\inc}(x,t) =0.
\end{align} 
We let the total wave field $u(x,t)$ be given by
\begin{align}
    u(x,t)=\sum\limits_{n=-\infty}^{\infty}v_n(x)\mathrm{e}^{-\mathrm{i}(\omega+n\Omega )t},
\end{align}
which is furnished by the assumption that $u(x,t) \mathrm{e}^{\mathrm{i} \omega t}$ is periodic with respect to $t$ with period $T$, where $\omega$ is the frequency of the incident wave. We note that the total wave field $u$ consists of the scattered wave field and the incident wave field, in particular,
\begin{align}
    u(x,t)=\begin{cases}
        u^{\mathrm{sc}}(x,t)+u^{\mathrm{in}}(x,t),&x\notin D,\\
        u^{\mathrm{sc}}(x,t),&x\in D.
    \end{cases}
\end{align}

Assuming that $u^{\mathrm{sc}}(x,t) \mathrm{e}^{\mathrm{i} \omega t}$ and $u^{\mathrm{in}}(x,t) \mathrm{e}^{\mathrm{i} \omega t}$ are periodic with respect to $t$ with period $T:=2\pi/\Omega$, we write their Fourier expansions with respect to $t$ as follows:  
\begin{equation}\label{def:us_Fourier}
    u^{\mathrm{sc}}(x,t)=\sum\limits_{n=-\infty}^{\infty}v^{\mathrm{sc}}_n(x)\mathrm{e}^{-\mathrm{i}(\omega+n\Omega )t},\quad u^{\mathrm{in}}(x,t)=\sum\limits_{n=-\infty}^{\infty}v^{\mathrm{in}}_n(x)\mathrm{e}^{-\mathrm{i}(\omega+n\Omega )t},\quad\forall\,x\in\mathbb{R},\,t\geq0.
\end{equation}

The governing equations for the wave scattering by the collection of resonators $D$ are given by \cite{jinghao_liora,erik_JCP}:
\begin{equation}\label{eq:dD_system_u} 
	\left\{
	\begin{array} {ll}
	\ds \frac{\p^2}{\p t^2}u^\s(x,t) - v_0^2\Delta u^\s(x,t) = 0, \quad &x\notin D, \\[1em]
	\ds \frac{\p }{\p t } \frac{1}{\kappa_i(t)} \frac{\p}{\p t}u^\s(x,t) - v_{\mathrm{r}}^2 \Delta u^\s(x,t) = 0, \quad &x\in D_i, \quad i=1,\ldots,N,\\[1em]
	\ds u^\s|_-(x_i^-,t) - u^\s|_+(x_i^-,t) = u^\inc|_-(x_i^-,t), \quad &i=1,\dots,N, \\[0.5em] 	
    \ds u^\s|_-(x_i^+,t) - u^\s|_+(x_i^+,t) = -u^{\inc}|_+(x_i^+,t), \quad &i=1,\dots,N, \\[0.5em] 
	\ds \frac{\p u^\s}{\p x }\bigg|_{+}(x_i^-,t) - \delta\frac{\p u^\s}{\p x }\bigg|_{-}(x_i^-,t) = \delta \frac{\p u^\inc}{\p x }\bigg|_{-}(x_i^-,t), &i=1,\dots,N, \\[1em]
	\ds  \frac{\p u^\s}{\p x }\bigg|_{-}(x_i^+,t) - \delta\frac{\p u^\s}{\p x }\bigg|_{+}(x_i^+,t) = \delta \frac{\p u^\inc}{\p x }\bigg|_{+}(x_i^+,t), &i=1,\dots,N, \\ [1em]
	u^\s \text{ is an outgoing wave,}
\end{array}		
\right.
\end{equation}
where we use the notation
\begin{align}
    \left.w\right|_{\pm}(x):=\lim_{s\rightarrow0,\,s>0}w(x\pm s).
\end{align}
\par 

We define the  resonant quasifrequencies as the set of frequencies $\omega_i$ with real parts in the first Brillouin zone $[-\Omega/2,\Omega/2)$ such that there is a non-zero (Bloch) solution $u_i(x,t)$ to \eqref{eq:dD_system_u} with $u^\textrm{in}(x,t) = 0$ and such that $u_i(x,t) \mathrm{e}^{\mathrm{i} \omega_i t}$ is periodic with respect to $t$ with period $T$. Such $\omega_i$ is a subwavelength resonant quasifrequency if the corresponding eigenmode is essentially supported in the subwavelength frequency regime, \textit{i.e.}, the infinite sum in \eqref{def:us_Fourier} can be approximated by a finite one with a number of terms $K \ll 1/\sqrt{\delta}$,  as $\delta \rightarrow 0$; see \cite{jinghao_liora,erik_JCP}. 

A practical characterisation of these resonant quasifrequencies can be obtained 
by decomposing $u^{\mathrm{sc}}$ into its Fourier modes. In fact, the governing equations posed on the modes $v_n^{\mathrm{sc}}(x)$ and $v_n^{\mathrm{in}}(x)$ of $u^{\mathrm{sc}}(x,t)$ and $u^{\mathrm{in}}(x,t)$ defined in \eqref{def:us_Fourier} can be derived to satisfy the following system of equations \cite{ammari2024scattering}:
\begin{equation}\label{eq:1DL_system}
	\left\{
	\begin{array} {ll}
	\ds  \frac{\mathrm{d}^2}{\mathrm{d}x^2}v^{\mathrm{sc}}_n+\frac{\rho_0(\omega+n\Omega)^2}{\kappa_0}v^{\mathrm{sc}}_n=0, \quad &x\notin D, \\[1em]
	\ds \frac{\mathrm{d}^2}{\mathrm{d}x^2}v_{n}^{\mathrm{sc}}+\frac{\rho_{\mathrm{r}}(\omega+n\Omega)^2}{\kappa_{\mathrm{r}}}v_{i,n}^{**}=0, \quad &x\in D_i, \quad i=1,\ldots,N,\\[1em]
	\ds \left.v^{\mathrm{sc}}_n\right|_{-}\left(x_i^-\right) -\left.v^{\mathrm{sc}}_n\right|_{+}\left(x_i^-\right)=v_n^{\inc}|_-(x_i^-), \quad &\forall\,i=1,\dots,N, \\[0.5em]
    \ds \left.v^{\mathrm{sc}}_n\right|_{-}\left(x_i^+\right) -\left.v^{\mathrm{sc}}_n\right|_{+}\left(x_i^+\right)=-v_n^{\inc}|_+(x_i^+), \quad &\forall\,i=1,\dots,N, \\[0.5em] 	
	\ds \left.\frac{\mathrm{d} v_{n}^{\mathrm{sc}}}{\mathrm{~d} x}\right|_{+}\left(x_i^{-}\right)-\left.\delta \frac{\mathrm{d} v^{\mathrm{sc}}_n}{\mathrm{d} x}\right|_{-}\left(x_i^{-}\right)=\left.\delta\frac{\mathrm{d}v_n^{\mathrm{in}}}{\mathrm{d}x}\right|_{-}(x_i^{-}), &  \forall\,i=1,\ldots,N, \\[1em]
    \ds \left.\frac{\mathrm{d} v_{n}^{\mathrm{sc}}}{\mathrm{~d} x}\right|_{-}\left(x_i^{+}\right)-\left.\delta \frac{\mathrm{d} v^{\mathrm{sc}}_n}{\mathrm{d} x}\right|_{+}\left(x_i^{+}\right)=\left.\delta\frac{\mathrm{d}v_n^{\mathrm{in}}}{\mathrm{d}x}\right|_{+}(x_i^{+}), &  \forall\,i=1,\ldots,N, \\[1em]
	\ds  \left( \frac{\mathrm{d}}{\mathrm{d} |x|} - \mathrm{i}
    \frac{(\omega + n \Omega)}  {v_0}  \right) v^{\mathrm{sc}}_n = 0, & 
    x \in (-\infty, x_1^-) \cup (x_N^+, + \infty),
\end{array}		
\right.
\end{equation}
where the functions $v_{i, n}^{* *}(x)$ are defined through the convolution
\begin{equation}\label{def:conv_v}
    v_{i,n}^{* *}(x)=\frac{1}{\omega+n \Omega} \sum_{m=-\infty}^{\infty} k_{i,m}(\omega+(n-m) \Omega) v_{n-m}(x),\quad\forall\,x\in D_i,
\end{equation}
with $k_{i, m}$ being the Fourier series coefficients of $1 / \kappa_i(t)$:
\begin{align*}
    \frac{1}{\kappa_i(t)} = \sum_{n=-M}^M k_{i, n}\mathrm{e}^{-\mathrm{i}n\Omega t}.
\end{align*}

We refer to \cite[Appendix A]{ammari2024scattering} for a detailed derivation of \eqref{eq:1DL_system}. In the following, we denote the wave numbers corresponding to the $n$-th mode inside and outside of the resonators by
\begin{align*}
    k^{(n)}_{\mathrm{r}}:=\frac{\omega+n\Omega}{v_{\mathrm{r}}},\quad k^{(n)}:=\frac{\omega+n\Omega}{v_0},
\end{align*} 
respectively. \par


\begin{rem}
While the incident frequency $\omega$ is real, the resonant frequencies $\omega_i$ of \eqref{eq:1DL_system} are, in general, complex. Note that, as shown in   \cite{hiltunen2024coupled}, there may be real resonances. At such frequencies (known in physics as lasing points \cite{hiltunen2025energybalanceopticaltheorem}), the scattering problem \eqref{eq:1DL_system} is ill-posed and we do not expect solutions to be unique. For the remainder of this paper, we assume that all the resonant frequencies $\omega_i$ are strictly non-real, so that there exists a unique solution to \eqref{eq:1DL_system} for real-valued $\omega$.
Note that under this assumption, the well-posedness of \eqref{eq:1DL_system} for $\omega$ real and of \eqref{eq:dD_system_u} are equivalent. 
\end{rem}

\par
The following proposition summarises the results derived in previous papers \cite{jinghao_liora,ammari2024scattering}:
\begin{prop}\label{prop:pre}
The Fourier modes $v^{\mathrm{sc}}_n(x)$ of the wave field $u^{\mathrm{sc}}(x,t)$ solving \eqref{eq:1DL_system} are given by
    \begin{align}
        v^{\mathrm{sc}}_n(x)=\begin{cases}
            \alpha_n^i\mathrm{e}^{\mathrm{i}k^{(n)}x}+\beta_n^i\mathrm{e}^{-\mathrm{i}k^{(n)}x},&\forall\,x\in\left(x_{i-1}^+,x_{i}^-\right),\\
            \sum\limits_{j=-\infty}^{\infty}\left(a_j^i\mathrm{e}^{\mathrm{i}\lambda_j^ix}+b_j^i\mathrm{e}^{-\mathrm{i}\lambda_j^ix}\right)f_n^{j,i},&\forall\,x\in\left(x_i^-,x_i^+\right),
        \end{cases}
    \end{align}
    where the coefficients $\alpha_n^i,\,\beta_n^i,\,a_j^i,\,b_j^i$ need to be determined for all $i=1,\dots,N,\,j,n\in\mathbb{Z}$. The eigenpairs $\left(\lambda_j^i,\boldsymbol{f}^{j,i}\right)_{i\in\mathbb{Z}}$ corresponding to the $i$-th resonator can be obtained as stated in \cite[Lemma 2.1]{ammari2024scattering}.
\end{prop}
\begin{rem}\label{rem:truncation}
    The wave field $u(x,t)$ has infinitely many modes $v_n(x)$, which is problematic for the numerical implementation of the equations. Thus, we approximate $u$ by a truncated Fourier series with $K\in\mathbb{N}$ \cite{jinghao_liora}:
    \begin{align}\label{eq:utrunc}
    u(x,t)\approx\sum\limits_{n=-K}^{K}v_n(x)\mathrm{e}^{-\mathrm{i}(\omega+n\Omega )t}.
\end{align}
Throughout the remainder of this work, we consider a truncation of \eqref{eq:1DL_system} where we restrict to $-K\leq n \leq K$ and similarly seek a truncated $u$ of the form \eqref{eq:utrunc}. This allows for numerical results, since otherwise, we would be dealing with an infinite coupled system. We note that the convergence theory and an estimation of the error introduced by  truncating the Fourier series of $u$ remain open problems.
\end{rem}

\section{Static Metamaterial}\label{sec:static_material}
In this section, we shall focus on metamaterials with static material parameters. We aim to obtain the homogenised governing equation in a manner equivalent to \cite{Hai_Habib}, but for a one-dimensional material. The method we introduce will be later generalised to the time-modulated case in \Cref{sec:timedep_material}. \par 
\subsection{Transfer and Scattering Matrices}
First we introduce the one-dimensional transfer matrix and the scattering matrix.
\begin{prop}
    Consider the resonator $D_i$, and let the wave field on the left- and right-hand sides of $D_i$ be given by $v^{\mathrm{sc}}(x)=\alpha^{i}\mathrm{e}^{\mathrm{i}kx}+\beta^i\mathrm{e}^{-\mathrm{i}kx}$ and $v^{\mathrm{sc}}(x)=\alpha^{i+1}\mathrm{e}^{\mathrm{i}kx}+\beta^{i+1}\mathrm{e}^{-\mathrm{i}kx}$, respectively. Then the transfer matrix $\Tilde{S}_i$ corresponding to $D_i$ is given by
    \begin{align}
        \begin{bmatrix}
            \alpha^{i+1}\\
            \beta^{i+1}
        \end{bmatrix} = \Tilde{S}_i\begin{bmatrix}
            \alpha^{i}\\
            \beta^{i}
        \end{bmatrix},
    \end{align}
    with coefficients given by
    \begin{align}\label{eq:Stildei_def}
    \Tilde{S}_i=\begin{bmatrix}
            \Tilde{a}_i & \Tilde{b}_i\\
            \Tilde{c}_i & \Tilde{d}_i
        \end{bmatrix},\quad \begin{cases}
            \Tilde{a}_i:=-\frac{\mathrm{e}^{\mathrm{i}\left(k(x_i^--x_i^+)-(x_i^++x_i^-)k_{\mathrm{r}}\right)}\left(\left(\delta^2k^2+k_{\mathrm{r}}^2\right)\left(\mathrm{e}^{2\mathrm{i}x^-k_{\mathrm{r}}}-\mathrm{e}^{2\mathrm{i}x^+\mathrm{k}_{\mathrm{r}}}\right)-2\delta k_{\mathrm{r}}^2\left(\mathrm{e}^{2\mathrm{i}x^-k_{\mathrm{r}}}+\mathrm{e}^{2\mathrm{i}x^+\mathrm{k}_{\mathrm{r}}}\right)\right)}{4\delta k k_{\mathrm{r}}}, \\[4pt]
            \Tilde{b}_i:=\frac{\mathrm{e}^{-\mathrm{i}(x_i^-+x_i^+)(k+k_{\mathrm{r}})}\left(\delta^2k^2-k_{\mathrm{r}}^2\right)\left(\mathrm{e}^{2\mathrm{i}x_i^-k_{\mathrm{r}}}-\mathrm{e}^{2\mathrm{i}x_i^+k_{\mathrm{r}}}\right)}{4\delta k k_{\mathrm{r}}}, \\[4pt]
            \Tilde{c}_i:=\frac{\mathrm{e}^{\mathrm{i}(x_i^-+x_i^+)(k-k_{\mathrm{r}})}\left(-\delta^2k^2-k_{\mathrm{r}}^2\right)\left(\mathrm{e}^{2\mathrm{i}x_i^-k_{\mathrm{r}}}-\mathrm{e}^{2\mathrm{i}x_i^+k_{\mathrm{r}}}\right)}{4\delta k k_{\mathrm{r}}}, \\[4pt]
            \Tilde{d}_i:=\frac{\mathrm{e}^{\mathrm{i}\left(k(x_i^--x_i^+)-(x_i^++x_i^-)k_{\mathrm{r}}\right)}\left(\left(\delta^2k^2+k_{\mathrm{r}}^2\right)\left(\mathrm{e}^{2\mathrm{i}x^-k_{\mathrm{r}}}-\mathrm{e}^{2\mathrm{i}x^+\mathrm{k}_{\mathrm{r}}}\right)-2\delta k_{\mathrm{r}}^2\left(\mathrm{e}^{2\mathrm{i}x^-k_{\mathrm{r}}}+\mathrm{e}^{2\mathrm{i}x^+\mathrm{k}_{\mathrm{r}}}\right)\right)}{4\delta k k_{\mathrm{r}}}.
        \end{cases}
    \end{align}
    Furthermore, the scattering matrix $S_i$ is defined by
    \begin{align}\label{def:scattering_matrix}
        \begin{bmatrix}
            \alpha^{i+1} \\
            \beta^{i}
        \end{bmatrix}=S_i\begin{bmatrix}
            \alpha^{i} \\
            \beta^{i+1}
        \end{bmatrix},
    \end{align}
    with coefficients given by
    \begin{align}
        S_i=\begin{bmatrix}
            a_i & b_i \\
            c_i & d_i
        \end{bmatrix}, \quad \begin{cases}
            a_i:= \Tilde{a}_i-\frac{\Tilde{b}_i\Tilde{c}_i}{\Tilde{d}_i},\\
            b_i:= \frac{\Tilde{b}_i}{\Tilde{d}_i},\\
            c_i:= -\frac{\Tilde{c}_i}{\Tilde{d}_i},\\
            d_i:= \frac{1}{\Tilde{d}_i}.
        \end{cases}
    \end{align}
\end{prop}
\begin{proof}
If we let the interior wave field be given by $v^{\mathrm{sc}}(x) = a\mathrm{e}^{\mathrm{i}k_{\mathrm{r}}x}+b\mathrm{e}^{-\mathrm{i}k_{\mathrm{r}}x}, \ x\in D_i$, the continuity and transmission conditions give us
    \begin{align}
        &\begin{cases}\label{eq:CC1_CC2}
            \alpha^i\mathrm{e}^{\mathrm{i}kx_i^-}+\beta^i\mathrm{e}^{-\mathrm{i}kx_i^-}=a\mathrm{e}^{\mathrm{i}k_{\mathrm{r}}x_i^-}+b\mathrm{e}^{-\mathrm{i}k_{\mathrm{r}}x_i^-},\\
            \alpha^{i+1}\mathrm{e}^{\mathrm{i}kx_i^+}+\beta^{i+1}\mathrm{e}^{-\mathrm{i}kx_i^+}=a\mathrm{e}^{\mathrm{i}k_{\mathrm{r}}x_i^+}+b\mathrm{e}^{-\mathrm{i}k_{\mathrm{r}}x_i^+},
        \end{cases}\\ 
        &\begin{cases}\label{eq:TC1_TC2}
            \delta k\left(\alpha^i\mathrm{e}^{\mathrm{i}kx_i^-}-\beta^i\mathrm{e}^{-\mathrm{i}kx_i^-}\right)=k_{\mathrm{r}}\left(a\mathrm{e}^{\mathrm{i}k_{\mathrm{r}}x_i^-}-b\mathrm{e}^{-\mathrm{i}k_{\mathrm{r}}x_i^-}\right),\\
            \delta k\left(\alpha^{i+1}\mathrm{e}^{\mathrm{i}kx_i^+}-\beta^{i+1}\mathrm{e}^{-\mathrm{i}kx_i^+}\right)=k_{\mathrm{r}}\left(a\mathrm{e}^{\mathrm{i}k_{\mathrm{r}}x_i^+}-b\mathrm{e}^{-\mathrm{i}k_{\mathrm{r}}x_i^+}\right),
        \end{cases}
    \end{align}
    respectively. By introducing the matrices
    \begin{align*}
        &P_i:=\begin{bmatrix}
            \mathrm{e}^{\mathrm{i}kx_i^+} & \mathrm{e}^{-\mathrm{i}kx_i^+} \\
            \mathrm{e}^{\mathrm{i}kx_i^+} & -\mathrm{e}^{-\mathrm{i}kx_i^+}
        \end{bmatrix},\quad V:=\begin{bmatrix}
            1 & 0 \\ 0 & \delta\frac{k}{k_{\mathrm{r}}}
        \end{bmatrix},\quad G_i:=\begin{bmatrix}
            \mathrm{e}^{\mathrm{i}k_{\mathrm{r}}x_i^+} & \mathrm{e}^{-\mathrm{i}k_{\mathrm{r}}x_i^+} \\
            \mathrm{e}^{\mathrm{i}k_{\mathrm{r}}x_i^+} & -\mathrm{e}^{-\mathrm{i}k_{\mathrm{r}}x_i^+}
        \end{bmatrix},\\
        &F_i:= \begin{bmatrix}
            \mathrm{e}^{\mathrm{i}k_{\mathrm{r}}x_i^-} & \mathrm{e}^{-\mathrm{i}k_{\mathrm{r}}x_i^-} \\
            \mathrm{e}^{\mathrm{i}k_{\mathrm{r}}x_i^-} & -\mathrm{e}^{-\mathrm{i}k_{\mathrm{r}}x_i^-}
        \end{bmatrix},\quad  Q_i:=\begin{bmatrix}
            \mathrm{e}^{\mathrm{i}kx_i^-} & \mathrm{e}^{-\mathrm{i}kx_i^-} \\
            \mathrm{e}^{\mathrm{i}kx_i^-} & -\mathrm{e}^{-\mathrm{i}kx_i^-}
        \end{bmatrix},
    \end{align*}
    and rewriting and grouping the equations \eqref{eq:CC1_CC2} and \eqref{eq:TC1_TC2}, we obtain
    \begin{align}
        \Tilde{S}_i:=P_i^{-1}V^{-1}G_iF_i^{-1}VQ_i,
    \end{align}
    which is exactly \eqref{eq:Stildei_def}. Lastly, it is a 
    straightforward task to obtain the expression of $S_i$ in terms of $\Tilde{S}_i$.
\end{proof}

With the explicit formula for $S_i$ at hand, we will derive a point-scattering approximation in the limit when the resonator length $\ell$ tends to zero. To fix the asymptotic regime, we will consider the subwavelength regime where $\omega = O(\ell)$ and $\delta = O(\ell)$ as $\ell \to 0$. In the remainder of this paper, we will denote the $n\times n$ identity matrix by $I_n$.  
\begin{lemma}\label{lemma:S_characterisation_l0}
    Let the $N$ resonators each be of length $\ell$ and centred around $z_i$, and set $\delta=\gamma\ell$ and $\omega=\mu\ell$, for fixed $\gamma,\,\mu>0$ independent of $\ell$. As $\ell\to0$, the following holds:
    \begin{align}\label{eq:S_characterisation_l0}
        S_i=I_2+g\begin{bmatrix}
            1 & \mathrm{e}^{-2\mathrm{i}kz_i} \\ \mathrm{e}^{2\mathrm{i}kz_i} & 1
        \end{bmatrix}+O(\ell^2),\quad \text{with}\quad g:=\frac{\mathrm{i}\ell\mu v}{2\gamma v_{\mathrm{r}}^2}.
    \end{align}
\end{lemma}
\begin{proof}
    The above form of $S_i$ can be obtained by expanding the coefficients $a_i,\,b_i,\,c_i,\,d_i$ in $\ell$ and omitting higher-order terms.
\end{proof}
\begin{rem}
    By plugging in \eqref{eq:S_characterisation_l0} into \eqref{def:scattering_matrix}, we obtain for sufficiently small $\ell>0$
    \begin{align}
        \begin{bmatrix}
            \Tilde{\alpha}^{i+1} \\ \Tilde{\beta}^i
        \end{bmatrix} = g\begin{bmatrix}
            1 & \mathrm{e}^{-2\mathrm{i}kz_i} \\ \mathrm{e}^{2\mathrm{i}kz_i} & 1
        \end{bmatrix}\begin{bmatrix}
            \alpha^i \\ \beta^{i+1}
        \end{bmatrix}+O(\ell^2),
    \end{align}
    where
    \begin{align}
        \Tilde{\alpha}^{i+1}:=\alpha^{i+1}-\alpha^i \quad \Tilde{\beta}^i:=\beta^i-\beta^{i+1}.
    \end{align}
\end{rem}
The above lemma leads to an expression of the scattered wave field using the one-dimensional Helmholtz Green's function
\begin{align}\label{def:greens_fct}
    G^k(x-y):=\frac{\mathrm{e}^{\mathrm{i}k|x-y|}}{2\mathrm{i}k}.
\end{align}
Let the total wave field $u(x,t)=\mathrm{e}^{-\mathrm{i}\omega t}v(x)$ be given by
\begin{align}
    v(x)=\begin{cases}
        v^{\mathrm{sc}}(x)+v^{\mathrm{in}}(x),&x\notin D,\\
        v^{\mathrm{sc}}(x),&x\in D,
    \end{cases}
\end{align}
where $v^{\mathrm{in}}(x):=\alpha^1\mathrm{e}^{\mathrm{i}kx}+\beta^{N+1}\mathrm{e}^{-\mathrm{i}kx}$ is the overall incident wave on the system of $N$ resonators. Moreover, we denote the incident wave field on $D_i$ by $v^{\mathrm{in}}_i(x)$ and the scattered wave field of $D_i$ by $v^{\mathrm{sc}}_i(x)$. Here, $v^{\mathrm{in}}$ is the overall incident field, while $v^{\mathrm{in}}_i$ is the field impinging on the $i$-th resonator. In other words,  $v^{\mathrm{in}}_i$ is given by the sum of the overall incident field and the scattered fields from the other resonators. The incident and scattered wave field of $D_i$ are given by
\begin{align}\label{def:vinv}
    v_i^{\mathrm{in}}(x):=v^{\mathrm{in}}(x)+\sum\limits_{j\neq i}v^{\mathrm{sc}}_j(x),\quad v_i^{\mathrm{sc}}(x):=\begin{cases}
        \Tilde{\beta}^i\mathrm{e}^{-\mathrm{i}kx},&x<z_i,\\
        \Tilde{\alpha}^i\mathrm{e}^{\mathrm{i}kx},&x>z_i.
    \end{cases}
\end{align}
Using the definition of the Green's function \eqref{def:greens_fct} one obtains the following characterisation:
\begin{align}\label{eq:vsci_greensfct}
    v^{\mathrm{sc}}_i(x)=\Tilde{g}G^k(x-z_i)v^{\mathrm{in}}_i(z_i),
\end{align}
where $\Tilde{g}:=2\mathrm{i}kg$, with $g$ as defined in Lemma \ref{lemma:S_characterisation_l0}. Equation \eqref{eq:vsci_greensfct} is a point-scattering approximation for a static resonator, which characterises the scattered wave from $D_i$ in terms of the Green's function. This formula sets the ground to derive the homogenised equation equivalently as in \cite{Hai_Habib}.
\subsection{Homogenised Equation}
In this section, we shall obtain the homogenised equation using our previous results. First, we substitute the scattered wave field of $D_i$ \eqref{eq:vsci_greensfct} into the incident wave field \eqref{def:vinv} of $D_i$ and evaluate at $z_i$:
\begin{align}\label{eq:vinsyst}
    v^{\mathrm{in}}_i(z_i)=v^{\mathrm{in}}(z_i)+\sum\limits_{j\neq i}\Tilde{g}G^k(z_i-z_j)v^{\mathrm{in}}_j(z_j).
\end{align}
Equation \eqref{eq:vinsyst} is a linear system of equations, \textit{i.e.}, we can define the vectors $\boldsymbol{x},\,\boldsymbol{b}\in\mathbb{C}^N$ and the matrix $T\in\mathbb{C}^{N\times N}$ with entries
\begin{align}
    x_i:=v^{\mathrm{in}}_i(z_i),\quad b_i:=v^{\mathrm{in}}(z_i),\quad T_{ij}:=\begin{cases}
        \Tilde{g}G^k(z_i-z_j),&i\neq j,\\
        0,&i=j,
    \end{cases}
\end{align}
then \eqref{eq:vinsyst} leads to the following:
\begin{align}\label{eq:xTb}
    \boldsymbol{x}-T\boldsymbol{x}=\boldsymbol{b}.
\end{align}
In order to homogenise our system, we must assume that while the number of resonators grows, they are still well-separated, which means that their size simultaneously decreases. In other words, we consider an asymptotic regime where the volume fraction of resonators is fixed.
\begin{ass}\label{ass:1}
Assume that there exists a constant $\Lambda>0$ such that $N\ell=\Lambda$, for all $N,\,\ell>0$. 
\end{ass}
Next, we have
\begin{align}\label{def:beta}
    \Tilde{g}=\frac{\mathrm{i}\ell\mu v}{2\gamma v_{\mathrm{r}}^2}\frac{2\mathrm{i}\mu \ell}{v}=-\frac{1}{N}\frac{\mu^2\ell}{\gamma v_{\mathrm{r}}^2}\Lambda=:\frac{1}{N}\beta\Lambda.
\end{align}
Hence, we can conclude that
\begin{align}
    v^{\mathrm{in}}_i(z_i)=v^{\mathrm{in}}(z_i)+\frac{1}{N}\sum\limits_{j\neq i}\beta\Lambda G^k(z_i-z_j)v^{\mathrm{in}}_j(z_j).
\end{align}
Following \cite{Hai_Habib}, we introduce the limiting density $\Tilde{V}(x)$ of resonators.
\begin{ass}\label{ass:density}
    Let $A\subset\mathcal{U}$ be a measurable set, we define 
    \begin{align}
        \Theta^N(A):=\frac{1}{N}\times\{\text{number of points } z_i \in A\}.
    \end{align}
    Assume that there exists $\Tilde{V}\in L^{\infty}(\mathcal{U})$ such that
    \begin{align}\label{eq:Vtilde}
        \Theta^N(A)\to\int_A\Tilde{V}(x)\,\mathrm{d}x\quad\text{as }N\to\infty.
    \end{align}
\end{ass}
We define the $L^{\infty}$-function 
$$V(x):= \left\{ \begin{array}{l} \beta\Lambda\Tilde{V}(x), \quad x \in \mathcal{U},\\ 
0, \quad x \in \mathbb{R} \setminus \mathcal{U}.
\end{array}
\right.
$$

The following assumption on the regularity of the distribution of the centres of the resonators is feasible based on \cite[Lemma 4.1]{Hai_Habib}.
\begin{ass}\label{ass:riemannsum_limit}
    Let $k$ be in a fixed neighbourhood of zero. We assume that for any $f\in C^{0,\alpha}(\mathcal{U})$
    \begin{align}
        \max_{1\leq i\leq N}\left|\frac{1}{N}\sum\limits_{j\neq i}\beta\Lambda G^k(z_i-z_j)f(z_j)-\int_{\mathcal{U}}G^k(z_i-y)V(y)f(y)\,\mathrm{d}y\right|\leq\frac{C}{N^{\alpha/3}}||f||_{C^{0,\alpha}(\mathcal{U})},
    \end{align}
    for some constants $0<\alpha<1$ and $C>0$, independent of $k$.
\end{ass}
We now define the operator $\mathcal{T}:C^{0,\alpha}(\mathcal{U})\to C^{0,\alpha}(\mathcal{U})$ by
\begin{align}
    \mathcal{T}f(x):=\int_{\mathcal{U}}G^k(x-y)V(y)f(y)\,\mathrm{d}y.
\end{align}
Analogously as in \cite{Hai_Habib}, this operator can be seen as the continuum limit of the matrix operator $T$. Since we define $\mathcal{T}$ in an equivalent manner as in the higher-dimensional case, the properties proven in \cite[Lemma 4.2]{Hai_Habib} still hold true in one dimension.\par 
Let $v^i\in C^{0,\alpha}(\mathcal{U})$, if $\psi$ is the unique solution of $\psi-\mathcal{T}\psi=v^i$, then
\begin{align}\label{eq:homogeq_static}
    \left(\frac{\mathrm{d}^2}{\mathrm{d}x^2}+k^2\right)\psi-V\psi=0\quad\text{in }\mathbb{R},
\end{align}
which is exactly the homogenised equation. Similarly as in \cite{Hai_Habib},  it can be proven that $\psi$ is the continuum limit of the solution $\boldsymbol{x}$ to \eqref{eq:xTb}.

\begin{rem}
Note that, in general, \eqref{eq:homogeq_static} may not be well-posed since $k^2-V$ may change sign inside $\mathcal{U}$. The effective potential $V$ may also be zero in some parts of $\mathcal{U}$ depending on the distribution of the centers of the resonators described by $\Tilde{V}$. This makes the study of the existence and uniqueness of a solution 
to \eqref{eq:homogeq_static} a delicate open problem. Moreover, note that if $k^2-V$ is negative, then the effective medium is dissipative. 
    \label{rem:well-posed}
\end{rem}

\section{Time-dependent Metamaterial}\label{sec:timedep_material}
In this section, we focus on time-dependent materials and obtain a homogenised equation. For that, we follow an equivalent approach as in Section \ref{sec:static_material}. Note that, as explained in
Remark \ref{rem:truncation}, we work with a truncated system with Fourier series of length $K$.\par 
Our aim is first to define the transfer and scattering matrix for time-dependent materials. To do so, we write out the continuity conditions and transmission conditions:
\begin{equation}
    \begin{cases}\label{eq:tdep_CC}
        \alpha^i_n\mathrm{e}^{\mathrm{i}k^{(n)}x_i^-}+\beta^i_n\mathrm{e}^{-\mathrm{i}k^{(n)}x_i^-}=\sum\limits_{j=-K}^{K}\left(a_j^i\mathrm{e}^{\mathrm{i}\lambda_j^ix_i^-}+b_j^i\mathrm{e}^{-\mathrm{i}\lambda_j^ix_i^-}\right)f_n^{j,i},\\
        \alpha^{i+1}_n\mathrm{e}^{\mathrm{i}k^{(n)}x_i^+}+\beta^{i+1}_n\mathrm{e}^{-\mathrm{i}k^{(n)}x_i^+}=\sum\limits_{j=-K}^{K}\left(a_j^i\mathrm{e}^{\mathrm{i}\lambda_j^ix_i^+}+b_j^i\mathrm{e}^{-\mathrm{i}\lambda_j^ix_i^+}\right)f_n^{j,i},
    \end{cases}
    \end{equation}
    and
    \begin{equation}
    \begin{cases}\label{eq:tdep_TC}
        \delta k^{(n)}\left(\alpha^i_n\mathrm{e}^{\mathrm{i}k^{(n)}x_i^-}-\beta^i_n\mathrm{e}^{-\mathrm{i}k^{(n)}x_i^-}\right)=\sum\limits_{j=-K}^{K}\left(a_j^i\mathrm{e}^{\mathrm{i}\lambda_j^ix_i^-}-b_j^i\mathrm{e}^{-\mathrm{i}\lambda_j^ix_i^-}\right)\lambda_j^if_n^{j,i},\\
        \delta k^{(n)}\left(\alpha^{i+1}_n\mathrm{e}^{\mathrm{i}k^{(n)}x_i^+}-\beta^{i+1}_n\mathrm{e}^{-\mathrm{i}k^{(n)}x_i^+}\right)=\sum\limits_{j=-K}^{K}\left(a_j^i\mathrm{e}^{\mathrm{i}\lambda_j^ix_i^+}-b_j^i\mathrm{e}^{-\mathrm{i}\lambda_j^ix_i^+}\right)\lambda_j^if_n^{j,i},
    \end{cases}
\end{equation}
for all $n=-K,\dots,K$. Next, we define a matrix $F_i:=\left(f_n^{j,i}\right)_{n,j=-K}^{K}$ and its inverse $G_i:=F_i^{-1}$. We will rewrite the above four equations, which hold for all $n=-K,\dots,K$, into a single matrix equation. We introduce the following notation for the unknown coefficients:
\begin{align}
    \boldsymbol{\alpha}^i:=\begin{bmatrix}
        \alpha^i_{-K}\\\vdots\\\alpha_K^i
    \end{bmatrix},\quad \boldsymbol{\beta}^i:=\begin{bmatrix}
        \beta^i_{-K}\\\vdots\\\beta^i_K
    \end{bmatrix},\quad \boldsymbol{w}^i:=\begin{bmatrix}
            \boldsymbol{\alpha}^i\\\boldsymbol{\beta}^i
        \end{bmatrix},\quad \boldsymbol{v}_j^i:=\begin{bmatrix}
            a_j^i\\b_j^i
        \end{bmatrix},\quad \boldsymbol{v}^i:=\begin{bmatrix}
            \boldsymbol{v}_{-K}^i\\\vdots\\\boldsymbol{v}_K^i
        \end{bmatrix}.
\end{align}
For the sake of simplicity, we omit the indices $i$ in the following notation. Then we introduce the matrices
\begin{align}
    &E_{\pm,\mathrm{l}}:=\mathrm{diag}\left(\mathrm{e}^{\pm\mathrm{i}k^{(n)}x^-}\right)_{n=-K}^ {K},\,\, E_{\pm,\mathrm{r}}:=\mathrm{diag}\left(\mathrm{e}^{\pm\mathrm{i}k^{(n)}x^+}\right)_{n=-K}^ {K},\,\, K:=\mathrm{diag}\left(k^{(n)}\right)_{n=-K}^ {K},\\
    &T^{\pm}:=\mathrm{diag}\left(\left[\mathrm{e}^{\mathrm{i}\lambda_j^ix^{\pm}},\,\mathrm{e}^{-\mathrm{i}\lambda_j^ix^{\pm}}\right]\right)_{j=-K}^ {K},\quad T^{\pm}_{\lambda}:=\mathrm{diag}\left(\left[\lambda_j^i\mathrm{e}^{\mathrm{i}\lambda_j^ix^{\pm}},\,-\lambda_j^i\mathrm{e}^{-\mathrm{i}\lambda_j^ix^{\pm}}\right]\right)_{j=-K}^ {K},\\
    &\mathcal{M}^-:=\begin{bmatrix}
        GE_{+,\mathrm{l}} & GE_{-,\mathrm{l}}\\
        \delta GKE_{+,\mathrm{l}} & -\delta GKE_{-,\mathrm{l}}
    \end{bmatrix},\,\, \mathcal{T}^-:=\begin{bmatrix}
        T^-\\T^-_{\lambda}
    \end{bmatrix},\,\, \mathcal{M}^+:=\begin{bmatrix}
        E_{+,\mathrm{r}} & E_{-,\mathrm{r}}\\ \delta KE_{+,\mathrm{r}} & -\delta KE_{-,\mathrm{r}}
    \end{bmatrix},\,\,\mathcal{T}^+:=\begin{bmatrix}
        FT^+\\FT^+_{\lambda}
    \end{bmatrix}.\label{def:matrices_Stilde}
\end{align}
With this notation it can be proven that \eqref{eq:tdep_CC} and \eqref{eq:tdep_TC} are equivalent to
\begin{align}
    \mathcal{M}^-\boldsymbol{w}^i=\mathcal{T}\boldsymbol{v}^i,\quad \mathcal{M}^+\boldsymbol{w}^{i+1}=\mathcal{T}^+\boldsymbol{v}^i.
\end{align}
We have now transformed the four equations \eqref{eq:tdep_CC} and \eqref{eq:tdep_TC}, that are true for each $n$, into two linear systems of equations, which contain all the modes $n$. By combining the two equations into a single one, we arrive at the following proposition, which defines the time-dependent scattering matrix. We assume that $\omega$ is not a multiple of $\Omega$, whereby a direct calculation shows that $\mathcal{M}^+$ and $\mathcal{T}^-$ are invertible. 
\begin{prop}
    Consider the resonator $D_i$, and assume that $\omega \neq m\Omega$ for all $m\in \Z$. Let the wave field on the left side of $D_i$ be given by the modes $v_n(x)=\alpha^i_n\mathrm{e}^{\mathrm{i}k^{(n)}x}+\beta^i_n\mathrm{e}^{-\mathrm{i}k^{(n)}x}$ and the wave field on the right side of $D_i$ be given by the modes $v_n(x)=\alpha^{i+1}_n\mathrm{e}^{\mathrm{i}k^{(n)}x}+\beta^{i+1}_n\mathrm{e}^{-\mathrm{i}k^{(n)}x}$. Then the corresponding transfer matrix is given by
    \begin{align}
        \boldsymbol{w}^{i+1}=\Tilde{S}_i\boldsymbol{w}^i,\quad \Tilde{S}_i:=\left(\mathcal{M}^+\right)^{-1}\mathcal{T}^+\left(\mathcal{T}^-\right)^{-1}\mathcal{M}^-=:\begin{bmatrix}
            \Tilde{S}_{11} & \Tilde{S}_{12} \\
            \Tilde{S}_{21} & \Tilde{S}_{22}
        \end{bmatrix}\in\mathbb{C}^{2(2K+1)\times2(2K+1)},
    \end{align}
    where these matrices are defined in \eqref{def:matrices_Stilde}. Furthermore, the scattering matrix is defined by
    \begin{align}\label{def:Si_tdep}
        S_i:=\begin{bmatrix}
            S_{11} & S_{12} \\ S_{21} & S_{22}
        \end{bmatrix},\quad \begin{cases}
            S_{11}:=\Tilde{S}_{11}-\Tilde{S}_{12}\Tilde{S}_{22}^{-1}\Tilde{S}_{21},\\
            S_{12}:=\Tilde{S}_{12}\Tilde{S}_{22}^{-1},\\
            S_{21}:=-\Tilde{S}_{22}^{-1}\Tilde{S}_{21},\\
            S_{22}:=\Tilde{S}_{22}^{-1},
        \end{cases}
    \end{align}
    and it satisfies
    \begin{align}
        \begin{bmatrix}
            \boldsymbol{\alpha}^{i+1}\\\boldsymbol{\beta}^i
        \end{bmatrix}=S_i\begin{bmatrix}
            \boldsymbol{\alpha}^i\\\boldsymbol{\beta}^{i+1}
        \end{bmatrix}.
    \end{align}
\end{prop}
As in the previous section, we proceed by taking the limit $\ell \to 0$ and derive a point-scattering approximation equivalent to Lemma \ref{lemma:S_characterisation_l0}. We take the same asymptotic scaling as before:
\begin{equation}
\delta=\gamma\ell\quad \omega=\mu\ell, \quad \Omega=\xi\ell
\end{equation}
for fixed $\gamma,\,\mu,\,\xi>0$ independent of $\ell$.  Moreover, recall that $\lambda_j^i$ are the wave numbers in the $i$-th resonator, as in \Cref{prop:pre}. The coefficients $\lambda_j^i$ are given by the eigenvalues of the matrix defined in \cite[Lemma III.3]{jinghao_liora} and, under the chosen parameter regime, it is straight-forward to show that $\lambda_j^i$ is of order $\ell$. As $\ell \to 0$, we then define the coefficients $c_j^i$ through 
\begin{equation}
\lambda_j^i=c_j^i\ell.
\end{equation}
Note that the exact definition of $c_j^i$ depends on $\omega$ and $\Omega$. Although there is no closed-form expression for these coefficients, they can be efficiently computed as the eigenvalues of the temporal Sturm-Liouville problem in the interior of $D_i$. Following \cite[Section V]{hiltunen2024coupled}, the coefficients $\lambda_j^i$ can be phrased as the eigenvalues of a Sturm-Liouville problem in the temporal variable. By taking a Fourier transform of this eigenvalue problem, we obtain the equivalent characterisation given in \cite[Lemma III.3]{jinghao_liora}. In the remainder of this paper, we will denote the $n\times n$ identity matrix by $I_n$.  
\begin{lemma}\label{lemma:S_characterisation_l0_tdep}
    Let the $N$ resonators each be of length $\ell$ and centred around $z_i$, and set $\delta=\gamma\ell,\,\omega=\mu\ell,\,\Omega=\xi\ell,\,\lambda_j^i=c_j^i\ell$, for fixed $\gamma,\,\mu,\,\xi>0$ independent of $\ell$. Assume that $\omega \neq m\Omega$ for all $m\in \Z$. Then, as $\ell\to0$, the following holds:
    \begin{align}\label{def:asymptotic_Si_tdep}
        S_i=\begin{bmatrix}
            I_{2K+1} & 0 \\ 0 & I_{2K+1}
        \end{bmatrix}+\begin{bmatrix}
            g^i & 0 \\ 0 & g^i
        \end{bmatrix}\begin{bmatrix}
            I_{2K+1} & \mathrm{diag}\left(\mathrm{e}^{-2\mathrm{i}k^{(n)}z_i}\right)_{n=-K}^K \\ \mathrm{diag}\left(\mathrm{e}^{2\mathrm{i}k^{(n)}z_i}\right)_{n=-K}^K & I_{2K+1}
        \end{bmatrix} + O(\ell^2),
    \end{align}
    where $g^i\in\mathbb{C}^{(2K+1)\times(2K+1)}$ is given by
    \begin{align}\label{def:g_Acal_tdep}
        g^i:=\left(I_{2K+1}-\mathcal{A}_i\right)^{-1}\mathcal{A}_i,\quad \mathcal{A}_i:=\frac{\ell}{2\gamma}\mathrm{diag}\left(\frac{v}{\mu+n\xi}\right)_{n=-K}^{K}F\mathrm{diag}\left(\mathrm{i}(c_j^i)^2\right)_{j=-K}^{K}G.
    \end{align}
\end{lemma}
For the proof of Lemma \ref{lemma:S_characterisation_l0_tdep} we shall use the following well-known result from linear algebra:
\begin{lemma}\label{lemma:inverse_blockmatrix}
    The inverse of an invertible $2m\times2m$ matrix
    \begin{align}
        M=\begin{bmatrix}
            A & B\\ C & D
        \end{bmatrix},
    \end{align}
    with blocks $A,\,B,\,C,\,D\in\mathbb{C}^{m\times m}$ is given by
    \begin{align}
        M^{-1}=\begin{bmatrix}
            A^{-1}+A^{-1}BS^{-1}CA^{-1} & -A^{-1}BS^{-1} \\ -S^{-1}CA^{-1} & S^{-1}
        \end{bmatrix},\quad S:=D-CA^{-1}B,
    \end{align}
    provided that $A$ and $S$ are invertible.  Note that $S$ is called the Schur complement of $A$.
\end{lemma}
With the result of Lemma \ref{lemma:inverse_blockmatrix}, we can now prove Lemma \ref{lemma:S_characterisation_l0_tdep}.
\begin{proof}[Proof of Lemma \ref{lemma:S_characterisation_l0_tdep}]
    Recall that the transfer matrix is given by $\Tilde{S}_i:=\left(\mathcal{M}^+\right)^{-1}\mathcal{T}^+\left(\mathcal{T}^-\right)^{-1}\mathcal{M}^-$, as defined in \eqref{def:matrices_Stilde}. The proof will consist of simplifying each term appearing in $\Tilde{S}_i$ and then taking the limit $\ell\to 0$.\par
    \underline{Invert $\mathcal{M}^+$}: \\
    The Schur complement of the first block of $\mathcal{M}^+$ is given by
    \begin{align}
        S_{\mathcal{M}^+}=-\delta KE_{-,\mathrm{r}}-(\delta KE_{+,\mathrm{r}})\left(E_{+,\mathrm{r}}\right)^{-1}E_{-,\mathrm{r}}=-2\delta KE_{-,\mathrm{r}}\quad\implies\quad S_{\mathcal{M}^+}^{-1}=-\frac{1}{2\delta}E_{+,\mathrm{r}}K^{-1}.
    \end{align}
    Then, by Lemma \ref{lemma:inverse_blockmatrix}, the inverse of $\mathcal{M}^+$ is given by
    \begin{align}
        \left(\mathcal{M}^+\right)^{-1}=\frac{1}{2}\begin{bmatrix}
            E_{-,\mathrm{r}} & \frac{1}{\delta}E_{-,\mathrm{r}}K^{-1} \\ E_{+,\mathrm{r}} & -\frac{1}{\delta}E_{+,\mathrm{r}}K^{-1}
        \end{bmatrix}.
    \end{align}\par
    \underline{Invert $\mathcal{T}^-$}:\\
    The inverse of $\mathcal{T}^-$ can be computed to be given by
    \begin{align}
        \left(\mathcal{T}^-\right)^{-1}=\begin{bmatrix}
            \hat{T}&\hat{T}_{\lambda }
        \end{bmatrix},\quad \hat{T}:=\mathrm{diag}\left(\frac{1}{2}\begin{bmatrix}
            \mathrm{e}^{-\mathrm{i}\lambda_j^ix^-}\\\mathrm{e}^{\mathrm{i}\lambda_j^ix^-}
        \end{bmatrix}\right)_{j=-K}^K,\quad \hat{T}_\lambda:=\mathrm{diag}\left(\frac{1}{2\lambda_j^i}\begin{bmatrix}
            \mathrm{e}^{-\mathrm{i}\lambda_j^ix^-}\\-\mathrm{e}^{\mathrm{i}\lambda_j^ix^-}
        \end{bmatrix}\right)_{j=-K}^K.
    \end{align}\par 
    \underline{Multiply $\mathcal{T}^+$ by $\left(\mathcal{T}^-\right)^{-1}$}:\\
    Recall that $\mathcal{T}^+$ is given by
    \begin{align}
        \mathcal{T}^+=\begin{bmatrix}
            FT^+ \\ FT^+_{\lambda}
        \end{bmatrix}=\underbrace{\begin{bmatrix}
            F & 0 \\ 0 & F
        \end{bmatrix}}_{=:\mathcal{F}}\begin{bmatrix}
            T^+ \\ T^+_{\lambda}
        \end{bmatrix}.
    \end{align}
    Then one obtains
    \begin{align}
        \mathcal{T}^+\left(\mathcal{T}^-\right)^{-1}=\mathcal{F}\begin{bmatrix}
            T^+\hat{T} & T^+\hat{T}_{\lambda} \\ T^+_{\lambda}\hat{T} & T^+_{\lambda}\hat{T}_{\lambda}
        \end{bmatrix}.
    \end{align}\par
    \underline{Simplify $\mathcal{M}^-$}:\\
    The matrix $\mathcal{M}^-$ can be written as
    \begin{align}
        \mathcal{M}^-=\underbrace{\begin{bmatrix}
            G & 0 \\ 0 & G
        \end{bmatrix}}_{=:\mathcal{G}}\begin{bmatrix}
            E_{-,\mathrm{l}} & E_{+,\mathrm{l}} \\ \delta KE_{-,\mathrm{l}}& -\delta KE_{+,\mathrm{l}}
        \end{bmatrix}.
    \end{align}\par
    \underline{Compute $\Tilde{S}_i$}:\\
    Finally, the transfer matrix can be written out as
    \begin{align}
        \Tilde{S}_i=\frac{1}{2}\begin{bmatrix}
            E_{-,\mathrm{r}} & \frac{1}{\delta}E_{-,\mathrm{r}}K^{-1} \\ E_{+,\mathrm{r}} & -\frac{1}{\delta}E_{+,\mathrm{r}}K^{-1}
        \end{bmatrix}\mathcal{F}\begin{bmatrix}
            T^+\hat{T} & T^+\hat{T}_{\lambda} \\ T^+_{\lambda}\hat{T} & T^+_{\lambda}\hat{T}_{\lambda}
        \end{bmatrix}\mathcal{G}\begin{bmatrix}
            E_{-,\mathrm{l}} & E_{+,\mathrm{l}} \\ \delta KE_{-,\mathrm{l}} & -\delta KE_{+,\mathrm{l}}
        \end{bmatrix},
    \end{align}
    and through some calculations one can achieve the following explicit definitions of the four submatrices of $\Tilde{S}_i$:
    {\small
    \begin{align*}
        &\Tilde{S}_{11}=\frac{1}{2}\left(\left(E_{-,\mathrm{r}}F(T^+\hat{T})+\frac{1}{\delta}E_{-,\mathrm{r}}K^{-1}F(T^+_{\lambda}\hat{T})\right)GE_{+,\mathrm{l}}+\left(\delta E_{-,\mathrm{r}}F(T^+\hat{T}_{\lambda})+E_{-,\mathrm{r}}K^{-1}F(T^+_{\lambda}\hat{T}_{\lambda})\right)GKE_{+,\mathrm{l}}\right),\\
        &\Tilde{S}_{12}=\frac{1}{2}\left(\left(E_{-,\mathrm{r}}F(T^+\hat{T})+\frac{1}{\delta}E_{-,\mathrm{r}}K^{-1}F(T^+_{\lambda}\hat{T})\right)GE_{-,\mathrm{l}}-\left(\delta E_{-,\mathrm{r}}F(T^+\hat{T}_{\lambda})+E_{-,\mathrm{r}}K^{-1}F(T^+_{\lambda}\hat{T}_{\lambda})\right)GKE_{-,\mathrm{l}}\right),\\
        &\Tilde{S}_{21}=\frac{1}{2}\left(\left(E_{+,\mathrm{r}}F(T^+\hat{T})-\frac{1}{\delta}E_{+,\mathrm{r}}K^{-1}F(T^+_{\lambda}\hat{T})\right)GE_{+,\mathrm{l}}+\left(\delta E_{+,\mathrm{r}}F(T^+\hat{T}_{\lambda})-E_{+,\mathrm{r}}K^{-1}F(T^+_{\lambda}\hat{T}_{\lambda})\right)GKE_{+,\mathrm{l}}\right),\\
        &\Tilde{S}_{22}=\frac{1}{2}\left(\left(E_{+,\mathrm{r}}F(T^+\hat{T})-\frac{1}{\delta}E_{+,\mathrm{r}}K^{-1}F(T^+_{\lambda}\hat{T})\right)GE_{-,\mathrm{l}}-\left(\delta E_{+,\mathrm{r}}F(T^+\hat{T}_{\lambda})-E_{+,\mathrm{r}}K^{-1}F(T^+_{\lambda}\hat{T}_{\lambda})\right)GKE_{-,\mathrm{l}}\right).
    \end{align*}}\par
    \underline{Let $\ell\to0$}:\\
    We now want to compute the leading order terms of $\Tilde{S}_{ij}$, as $\ell\to0$. One can prove
    \begin{align*}
        &T^+\hat{T}=I_{2K+1}+O(\ell^4), 
        \quad T^+_{\lambda}\hat{T}=\mathrm{diag}\left(\mathrm{i}(c_j^i)^2\right)_{j=-K}^K\ell^3+O(\ell^4),\quad T^+_\lambda \hat{T}_{\lambda}=I_{2K+1}+O(\ell^4),
    \end{align*}
    and that $E_{\pm}=I_{2K+1}+O(\ell)$. Plugging these asymptotic formulas into $\Tilde{S}_{ij}$ and neglecting higher order terms, we obtain
    \begin{align}
        \begin{cases}
            \Tilde{S}_{11}=I_{2K+1}+\mathcal{A}_i+O(\ell^2),\\
            \Tilde{S}_{12}=\mathcal{A}_i+O(\ell^2),\\
            \Tilde{S}_{21}=-\mathcal{A}_i+O(\ell^2),\\
            \Tilde{S}_{22}=I_{2K+1}-\mathcal{A}_i+O(\ell^2),
        \end{cases}
    \end{align}
    with $\mathcal{A}_i$ defined by \eqref{def:g_Acal_tdep}. If we then define $g^i=\left(I_{2K+1}-\mathcal{A}_i\right)^{-1}\mathcal{A}_i$ and use the formulas \eqref{def:Si_tdep}, we obtain the desired result \eqref{def:asymptotic_Si_tdep}.
\end{proof}
We now want to use the asymptotic result of Lemma \ref{lemma:S_characterisation_l0_tdep} in order to characterise the total wave field as $\ell\to 0$. Let us separate the total wave field as a sum of incident and scattered fields
\begin{align}\label{eq:vn}
    v_n(x)=\begin{cases}
        v_n^{\mathrm{sc}}(x)+v_n^{\mathrm{in}}(x),&x\notin D,\\
        v_n^{\mathrm{sc}}(x),&x\in D,
    \end{cases}
\end{align}
where $v_n^{\mathrm{in}}(x):=\alpha^1_n\mathrm{e}^{\mathrm{i}k^{(n)}x}+\beta^{N+1}_n\mathrm{e}^{-\mathrm{i}k^{(n)}x}$. The scattered wave field emerging from $D_i$ is given by
\begin{align}\label{eq:vsc_in_tdep}
    v^{\mathrm{sc}}_{i,n}(x)=G^{k^{(n)}}(x-z_i)\sum\limits_{m=-K}^{K}\Tilde{g}^i_{nm}v_{i,m}^{\mathrm{in}}(z_i),
\end{align}
where $\Tilde{g}^i_{nm}:=2\mathrm{i}k^{(n)}g^i_{nm}$ and we use the notation $g^i:=\left(g^i_{nm}\right)_{m,n=-K}^K$. Here, $v_{i,m}^{\mathrm{in}}$ is the field that is incident to $D_i$, and the scattered field is defined piecewise by
\begin{align}
    v^{\mathrm{sc}}_n(x):=\sum\limits_{i=1}^Nv^{\mathrm{sc}}_{i,n}(x).
\end{align}
Furthermore, the wave field impinging on $D_i$ is given by
\begin{align}\label{eq:vin_in_tdep}
    v^{\mathrm{in}}_{i,n}(x)=v^{\mathrm{in}}_n(x)+\sum\limits_{j\neq i}v_{j,n}^{\mathrm{sc}}(x)=v^{\mathrm{in}}_n(x)+\sum\limits_{j\neq i}\sum\limits_{m=-K}^K\Tilde{g}^j_{nm}v^{\mathrm{in}}_{j,m}(z_j)G^{k^{(n)}}(x-z_j).
\end{align}
Note that the characterisations \eqref{eq:vsc_in_tdep} and \eqref{eq:vin_in_tdep} hold true only as $\ell\to 0$ and can be derived using the results proven in Lemma \ref{lemma:S_characterisation_l0_tdep}.
\subsection{Homogenised Model}
To obtain the time-dependent homogenised equation, we first evaluate \eqref{eq:vin_in_tdep} at each resonator $z_i$:
\begin{align}\label{eq:vinzi_in_tdep}
    v^{\mathrm{in}}_{i,n}(z_i)=v^{\mathrm{in}}_n(z_i)+\sum\limits_{j\neq i}v_{j,n}^{\mathrm{sc}}(z_i)=v^{\mathrm{in}}_n(z_i)+\sum\limits_{j\neq i}\sum\limits_{m=-K}^K\Tilde{g}^j_{nm}v^{\mathrm{in}}_{j,m}(z_j)G^{k^{(n)}}(z_i-z_j),
\end{align}
for all $i=1,\dots,N$ and $n=-K,\dots,K$. Now we define the vectors $\boldsymbol{x},\boldsymbol{b}\in\mathbb{C}^{N(2K+1)}$ with entries
\begin{align}
    \boldsymbol{x}=\begin{bmatrix}
        v^{\mathrm{in}}_{1,-K}(z_1)\\ \vdots \\ v^{\mathrm{in}}_{1,K}(z_1) \\ \vdots \\ v^{\mathrm{in}}_{N,-K}(z_N)\\ \vdots \\ v^{\mathrm{in}}_{N,K}(z_N)
    \end{bmatrix},\quad \boldsymbol{b}=\begin{bmatrix}
        v^{\mathrm{in}}_{-K}(z_1)\\ \vdots \\ v^{\mathrm{in}}_{K}(z_1) \\ \vdots \\ v^{\mathrm{in}}_{-K}(z_N)\\ \vdots \\ v^{\mathrm{in}}_{K}(z_N)
    \end{bmatrix}
\end{align}
and let the matrix $T\in\mathbb{C}^{N(2K+1)\times N(2K+1)}$ be defined by
\begin{align}
    T^{ij}_{nm}:=&\begin{cases}
        \Tilde{g}_{nm}^jG^{k^{(n)}}(z_i-z_j),&i\neq j,\\
        0,&i=j,
    \end{cases}\qquad T:=\begin{bmatrix}\begin{bmatrix}
T^{ij}_{nm}
    \end{bmatrix}_{n,m=-K}^K
    \end{bmatrix}_{i,j=1}^N.
\end{align}
This notation allows us to rewrite \eqref{eq:vinzi_in_tdep} as
\begin{align}\label{eq:xTb_tdep}
    \boldsymbol{x}-T\boldsymbol{x}=\boldsymbol{b},
\end{align}
which is of the same form as in the static case \eqref{eq:xTb}, but here the linear system has a block structure. As in the static case, we must pose some assumptions in order to derive the homogenised model.
\begin{ass} \label{ass:4}
    Assume that $\kappa_i(t)\equiv\kappa(t)$ is identical across all resonators. Therefore, $\lambda^i_j=c^i_j\ell\equiv c_j\ell$ and as a consequence, $g^i_{nm}\equiv g_{nm}$.
\end{ass}
We also impose the assumptions of the previous section: \Cref{ass:1} (constant volume fraction), \Cref{ass:density}  (limiting density $\Tilde{V}(x)$), and \Cref{ass:riemannsum_limit} (continuum limit approximation). This allows us to write
\begin{align} \label{def:bnm}
    \Tilde{g}_{nm}=2\mathrm{i}k^{(n)}g_{nm}=\frac{2\mathrm{i}}{N}\frac{\mu+n\xi}{v}\Lambda g_{nm}=:\frac{1}{N}\beta_{nm}\Lambda.
\end{align}
The definition of $\beta_{nm}$ was obtained in the same way as in the static case \eqref{def:beta}. Recall that $g^i_{nm}$ defined by \eqref{def:g_Acal_tdep} does not depend on $\ell$ and therefore $\beta^i_{nm}$ does not depend on $\ell$ nor $N$. Hence, 
\begin{align}\label{eq:pre_LSeq_sum}
    v^{\mathrm{in}}_{i,n}(z_i)=v^{\mathrm{in}}_n(z_i)+\frac{1}{N}\sum\limits_{j\neq i}\sum\limits_{m=-K}^K\beta_{nm}\Lambda v^{\mathrm{in}}_{i,n}(z_j) G^{k^{(n)}}(z_i-z_j).
\end{align}

First, we let 
\begin{equation} \label{def:vnm}
    V_{nm}(x):=\begin{cases} \beta_{nm}\Lambda\Tilde{V}(x), & x \in \mathcal{U},\\
    0, & x \in \mathbb{R} \setminus \mathcal{U},
    \end{cases}
\end{equation}
for $\Tilde{V}$ defined through \eqref{eq:Vtilde}. In view of the definitions \eqref{def:bnm} and \eqref{def:vnm} of $\beta_{nm}$ and $V_{nm}$, and the fact that $K$ is fixed, we can use Assumption \ref{ass:riemannsum_limit} to obtain the following result.

\begin{lemma}\label{lemma:riemannsum_conv}
    For any $f_m\in C^{0,\alpha}(\mathcal{U})$, $m=-K, \ldots, K,$ with $0<\alpha\leq1$,
    \begin{align}
        \max\limits_{1\leq i\leq N}\Bigg|\frac{1}{N}\sum\limits_{j\neq i}\sum\limits_{m=-K}^K\beta_{nm}\Lambda G^{k^{(n)}}(z_i-z_j)f_m(z_j)-&\int_{\mathcal{U}}\sum\limits_{m=-K}^KG^{k^{(n)}}(z_i-y)V_{nm}(y)f_m(y)\,\mathrm{d}y\Bigg|\nonumber\\
        &\leq\sum\limits_{m=-K}^{K}\frac{C_{nm}}{N^{\alpha/3}}||f_m||_{C^{0,\alpha}(\mathcal{U})},
    \end{align}
    for all $n=-K,\dots,K$. Here, $C_{nm}$ is independent of $N$ and $\{f_m\}_{m=-K}^{K}.$
\end{lemma}

Next, define the operators $\mathcal{T}^{nm}:C^{0,\alpha}(\mathcal{U})\to C^{0,\alpha}(\mathcal{U})$ by
\begin{align}
    \mathcal{T}^{nm}f(x)=\int_{\mathcal{U}}G^{k^{(n)}}(x-y)V_{nm}(y)f(y)\,\mathrm{d}y,\quad\forall\,m,n=-K,\dots,K.
\end{align}
By substituting the above defined integral operator $\mathcal{T}^{nm}$ into Lemma \ref{lemma:riemannsum_conv}, applying it to \eqref{eq:pre_LSeq_sum}, and neglecting the remainder term, we obtain the following system of coupled Lippmann-Schwinger equations:
\begin{align}\label{eq:system_LippmannSchwinger}
    \psi_n-\sum_{m=-K}^K\mathcal{T}^{nm}\psi_m=\psi^{\mathrm{in}},\quad\forall\,n=-K,\dots,K.
\end{align}
Then, applying the operator $\left(\frac{\mathrm{d}^2}{\mathrm{d}x^2}+\left(k^{(n)}\right)^2\right)$ leads to the coupled system of differential equations
\begin{equation}\label{eq:homog_eq_tdep}
	\left(\frac{\mathrm{d}^2}{\mathrm{d}x^2}+\left(k^{(n)}\right)^2\right)\psi_n(x) -  \sum\limits_{m=-K}^{K}V_{nm}(x)\psi_m(x)=0, \quad x \in \mathbb{R}, 
\end{equation}
which describes the effective properties of our system of time-modulated subwavelength resonators.

\begin{rem}\label{rem:totalfield_limit}
We first emphasise that the term $V_{nm}(x)$, which incorporates the homogenised effect of the time-modulated resonators, vanishes outside of $\mathcal{U}$. Moreover, in \eqref{eq:homog_eq_tdep}, both $\psi_n$ and its derivative are continuous across the boundary of  $\mathcal{U}$. Note also that ${\psi}_n(z_i)$ is the continuum limit of ${v}^{\mathrm{in}}_{i,n}$ as $N\to\infty$, which is proved in detail for the higher-dimensional static case in \cite{Hai_Habib}. Looking at \eqref{eq:vn} and \eqref{eq:vinzi_in_tdep}, the total field is given by $v^{\mathrm{in}}_{i,n}+v^{\mathrm{sc}}_{i,n}=v_n$ for all $i=1,\dots,N$. However, as $N\to\infty$, the scattered field of a single resonator $v_{i,n}^{\mathrm{sc}}$ becomes negligible. Thus, ${\psi}_n$ is the continuum limit of the total wave field $v_n$.
\end{rem}
\begin{cor}\label{cor:diffint}
   The system of homogenised equations \eqref{eq:homog_eq_tdep} can be rewritten into a single vector-valued equation given by
   \begin{align}\label{eq:hom_main}
       \left(\frac{\mathrm{d}^2}{\mathrm{d}x^2}+K^2\right)\boldsymbol{\psi}(x)-\mathcal{V}(x)\boldsymbol{\psi}(x)=\boldsymbol{0},
   \end{align}
   for {$x \in \mathbb{R}$}, where the vector $\boldsymbol{\psi}:=\left(\psi_n\right)_{n=-K}^K$ and the matrices $K$ and $\mathcal{V}$ are given by $K := \mathrm{diag}\left(k^{(n)}\right)_{n=-K}^K$ and $\mathcal{V}(x):=\left(V_{nm}(x)\right)_{m,n=-K}^K$.
\end{cor}

\begin{cor}\label{cor:LS}
    Assume the material parameters to be scaled as in Lemma \ref{lemma:S_characterisation_l0_tdep}. The system of coupled equations \eqref{eq:system_LippmannSchwinger} can then be rewritten into a single vector Lippmann-Schwinger equation given by
    \begin{align}\label{eq:vector_LippmannSchwinger}
        \boldsymbol{\psi}(x)=\boldsymbol{v}^{\mathrm{in}}(x)+\int_{\mathcal{U}}\mathcal{G}(x-y)\mathcal{V}(y)\boldsymbol{\psi}(y)\,\mathrm{d}y,\quad x \in \mathbb{R},
    \end{align}
    where 
    \begin{align}
        \mathcal{V}(x):=\left({V}_{nm}(x)\right)_{n,m=-K}^K,\quad\mathcal{G}(x):=\mathrm{diag}\left(G^{k^{(n)}}(x)\right)_{n=-K}^K.
    \end{align}
\end{cor}


\Cref{cor:diffint} and \Cref{cor:LS} are the main results of our homogenisation theory for systems of time-modulated subwavelength resonators. They provide two equivalent formulations of the effective field $\boldsymbol{\psi}$. Note that the homogenisation model described by \eqref{eq:hom_main} holds under Assumptions \eqref{ass:1}, \eqref{ass:density}, \eqref{ass:riemannsum_limit} and \eqref{ass:4}, and is valid in the regime where the operating frequency is of order $\omega =O(\delta)$. 

 Comparing the homogenised equation \eqref{eq:hom_main} to the one in the static case in \eqref{eq:homogeq_static}, we observe that the effective model in the time-modulated case is described by a system of coupled differential equations, rather than by a single differential equation as in the static case. This is due to the fact that time-modulations give rise to a family of coupled harmonics at frequencies 
$(\omega + n \Omega)/v_\mathrm{r}$. Note also that in both effective models, the scattering coefficients $g$ and $g_{nm}$ depend on the operating frequency $\omega$. Moreover, 
as in Remark \eqref{rem:well-posed}, positivity of the matrix induces a dissipation effect on components of the effective field $\boldsymbol{\psi}$.

\begin{rem}

    To justify the asymptotic regime used in this section, we recall the formulas for the two subwavelength resonant frequencies $\omega_0, \omega_1$ of a single resonator in one dimension proved in \cite{ammari2024scattering}:\begin{align}\label{eq:om1_formula}
        \omega_0=0,\quad \omega_1\approx -\frac{2\mathrm{i}(v_{\mathrm{r}})^2\delta}{\ell v_0 T}\int_0^T {\kappa(t)} \,\mathrm{d}t + O\left(\frac{\delta^{3/2}}{\ell}\right),
    \end{align}
    as $\delta \to 0$. In particular, we note that $\Re \left(\omega_1\right) = O(\delta^{3/2}/\ell)$. Choosing $\delta =O(\ell)$ yields $\Re\left(\omega_1\right) = O(\delta^{1/2})$. Now, if we consider an operating frequency $\omega$  of order $\Re\left(\omega_1\right)$, then the summand of \eqref{eq:pre_LSeq_sum} is of order one as $\delta \to0$. As $N\to \infty$, this sum cannot converge and hence it is not possible to obtain an effective medium theory. Therefore, to obtain an effective medium theory, we must work in a regime for which the scattering coefficient $g$ of a single resonator goes to zero as $\delta \to 0$. To achieve this, one simple choice is to take $\omega =O(\delta).$

\end{rem}

\begin{rem}
As in Remark \ref{rem:well-posed}, it is also worth emphasising that the homogenised model \eqref{eq:hom_main} may not have a unique solution.  Nevertheless, by exactly the same arguments as those in \cite{Hai_Habib}, we can prove that if the homogenised model is well-posed, then we have $L^2$-convergence of the solution of the scattering problem with $N$ resonators to the solution of \eqref{eq:hom_main} as $N$  goes to infinity. 
\end{rem}    

\begin{rem}\label{rem:higherdim}
    Note that in higher dimensions the characterisation \eqref{eq:vn} still holds true. However, the scattered wave field is expressed in terms of the single layer potential $\mathcal{S}^k_D[\psi]$ through \cite{Hai_Habib}
    \begin{align}
        v_{i,n}^{\mathrm{sc}}(x):=\mathcal{S}^{k^{(n)}}_{D_i}[\psi_{i,n}](x),
    \end{align}
    for some surface densities $\psi_{i,n}\in L^2(\partial D)$, for all $i=1,\dots,N$ and $n\in\mathbb{Z}$. Then we can proceed in an equivalent manner as in the one-dimensional case to obtain a linear system of equations of the form \eqref{eq:xTb_tdep}. With a set of assumptions similar to those stated here, but with a higher-dimensional notation as in \cite[Section 2]{Hai_Habib}, we would arrive at a result similar to the one stated in Corollary \ref{cor:diffint}.
\end{rem}
\section{Numerical Results}\label{sec:numerical_results}
We now present a standard Nyström numerical scheme to solve the system of Lippmann-Schwinger equations \eqref{eq:vector_LippmannSchwinger}. This scheme converges at the same rate as the quadrature rule considered for the numerical integration \cite{atkinson}. In our code, we use a trapezoidal quadrature method and hence our Nyström numerical scheme is of second order. Note that, in the case where we consider random distributions of the centers of the resonators, the Monte-Carlo Nyström numerical scheme developed in \cite{florian_MonteCarlo}  can be used.

Let $\left(y_i\right)_{i=1}^M$ be a set of $M$ uniformly distributed points in $\mathcal{U}$. Discretising the Lippmann-Schwinger equation yields
\begin{align}
    \boldsymbol{\psi}(y_i)=\boldsymbol{v}^{\mathrm{in}}(y_i)+\frac{|\mathcal{U}|}{M}\sum\limits_{j=1}^M\mathcal{V}(y_i,y_j)\boldsymbol{\psi}(y_j),
\end{align}
where $|\mathcal{U}|$ is the length of the interval $\mathcal{U}$. For the numerical solution presented in this section, we set $\Tilde{V}(x):=1/|\mathcal{U}|.$ Based on this equation, we define the following:
\begin{align}
    \boldsymbol{\psi}^M:=\begin{bmatrix}
        \boldsymbol{\psi}(y_1) \\ \vdots \\ \boldsymbol{\psi}(y_M)
    \end{bmatrix},\quad \boldsymbol{v}^{\mathrm{in},M}:=\begin{bmatrix}
        \boldsymbol{v}^{\mathrm{in}}(y_1) \\ \vdots \\ \boldsymbol{v}^{\mathrm{in}}(y_M)
    \end{bmatrix},\quad \mathbb{V}:=\left(\mathcal{V}(y_i,y_j)\right)_{i,j=1}^M.
\end{align}
These definitions allow us to write
\begin{align}\label{eq:numerical_scheme}
    \left(I_{M(2K+1)}-\frac{|\mathcal{U}|}{M}\mathbb{V}\right)\boldsymbol{\psi}^M=\boldsymbol{v}^{\mathrm{in},M}.
\end{align}
As pointed out in Remark \ref{rem:totalfield_limit}, $\boldsymbol{\psi}^M$ is the total wave field in the continuum limit.
\begin{figure}[H]
    \centering
    \includegraphics[width=1.0\textwidth]{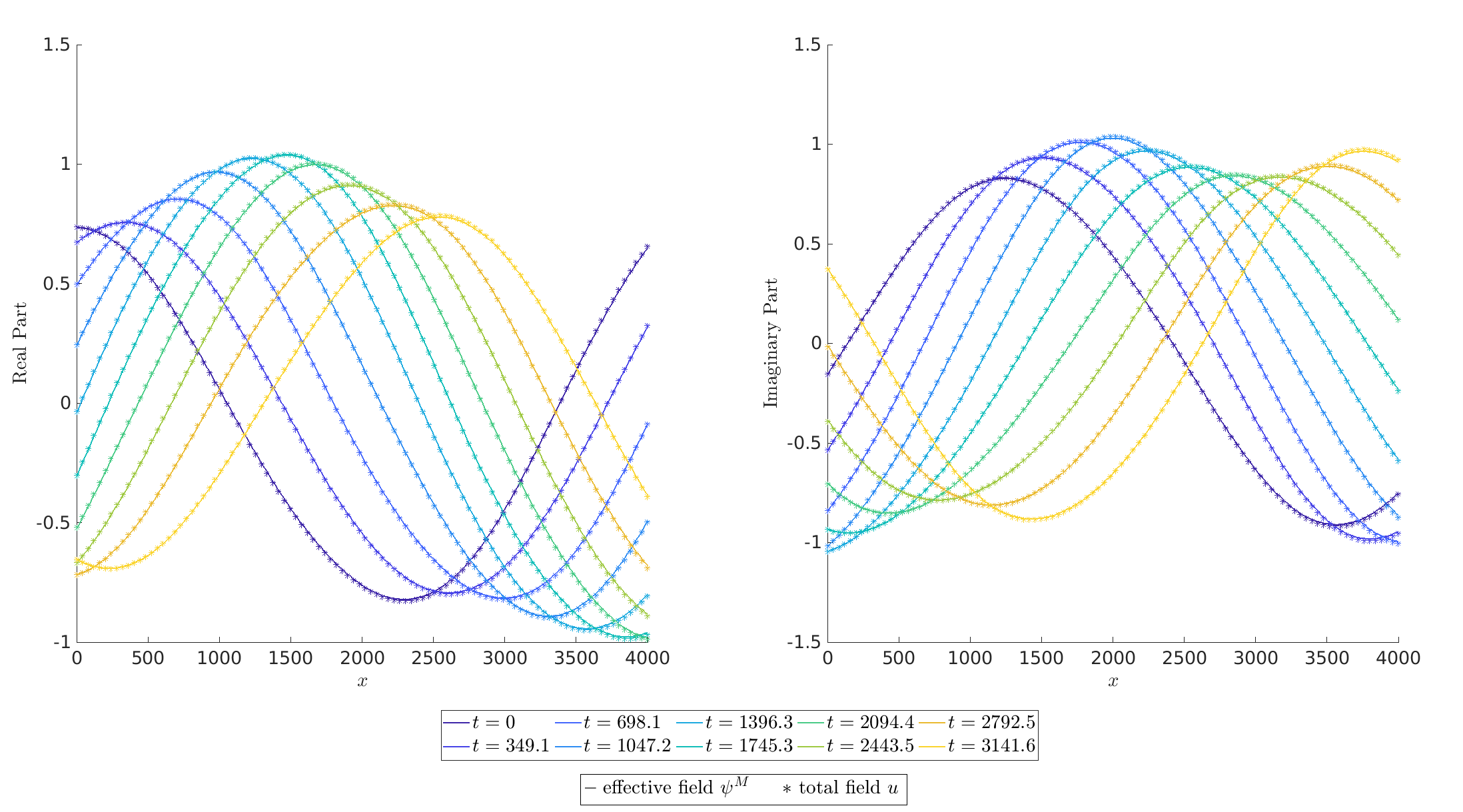}
    \caption{The effective field $\boldsymbol{\psi}^M$ (-) and the total field $u$ (*) plotted next to each other over $\mathcal{U}:=[0,4000]$ with $N=100$ resonators each of length $\ell=0.01$, thus $\Lambda=1$. We set $\gamma=0.05,\,\mu=0.11,\,\xi=0.2,\,K=4,\,\varepsilon_{\kappa}=0.4$. We evaluate the solution at ten time-steps in $[0,T]$, for $T=2\pi/\Omega=3141.59$. For the definition of \eqref{eq:numerical_scheme} we choose a mesh of size $M=200$.}\label{fig:effective_vs_true_timedep}
\end{figure}
\begin{figure}[H]
    \centering
    \includegraphics[width=0.8\textwidth]{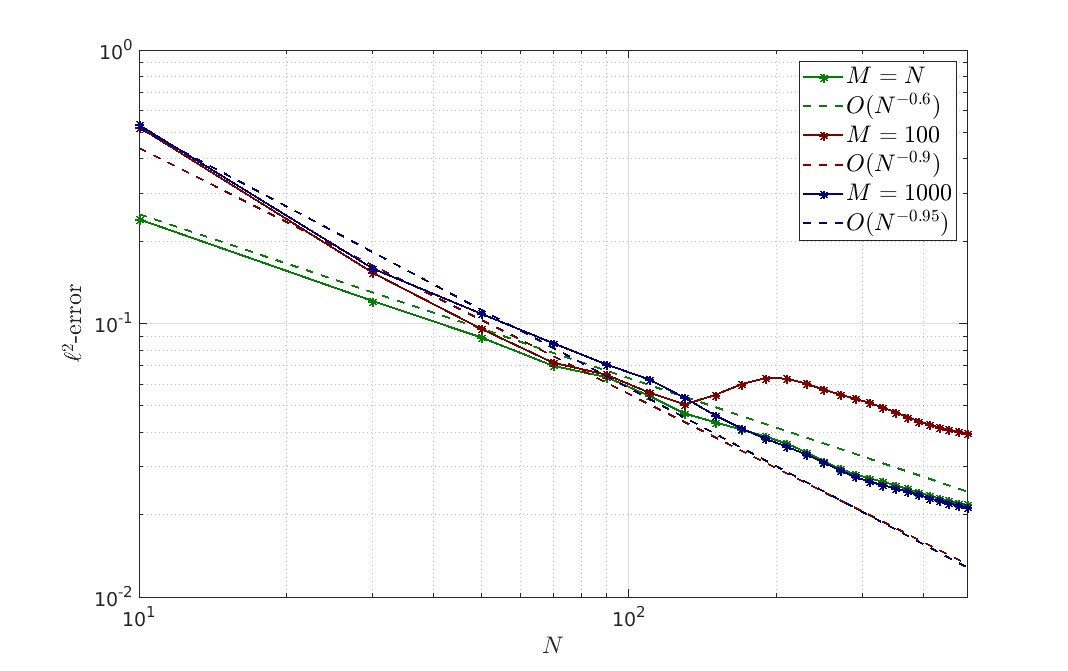}
    \caption{The $\ell^2$-norm of the error between the effective field $\boldsymbol{\psi}^M$ and the total field $u$ at time $t=0$. For the numerical solution we compare $M=N$, $M=100$ and $M=1000$. The resonators are evenly distributed inside $\mathcal{U}=[0,4000]$ and each of length $\ell=1/N$. We set $\gamma=0.05,\,\mu=0.11,\,\xi=0.2,\,K=4,\,\varepsilon_{\kappa}=0.4$.}\label{fig:error_plot}
\end{figure}

The bold line in Figure \ref{fig:effective_vs_true_timedep} shows the effective field $\boldsymbol{\psi}^M$ as a function of $x$ evaluated at different time-steps inside the interval $[0,T]$. For comparison, we plot the evaluation of the total wave field obtained using the scheme obtained in \cite{ammari2024scattering}, marked by stars inside the plot. As analytically proved, the effective field and the total field agree for $\ell\ll1$ and large $N$. The numerical results in Figure \ref{fig:effective_vs_true_timedep} show that the effective field is still quasi-periodic with quasi-periodicity $\omega$.\par
\begin{figure}[H]
    \centering
    \includegraphics[width=0.8\textwidth]{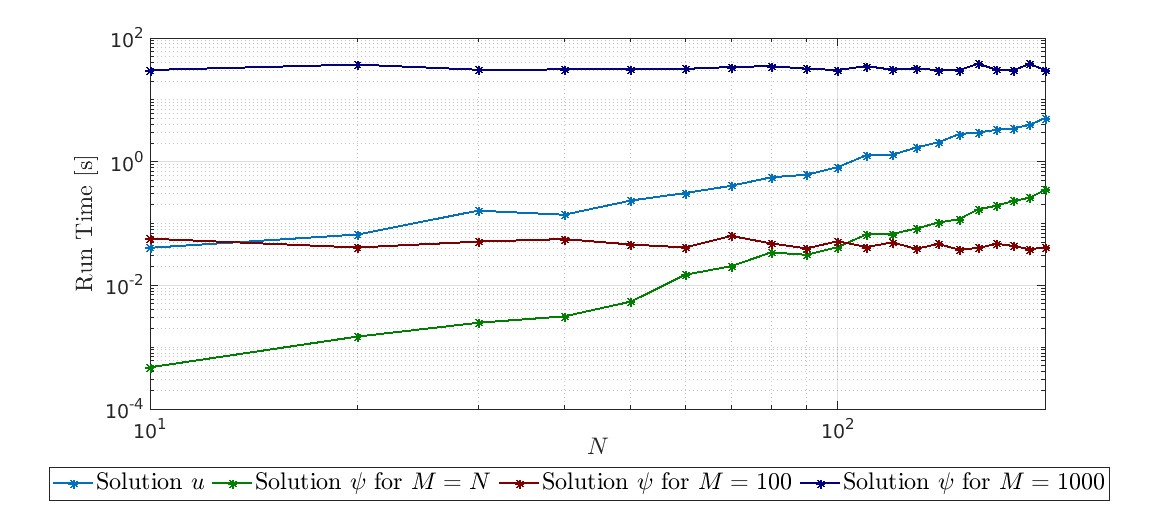}
    \caption{The run time of the computation of $\boldsymbol{\psi}^M$ solving \eqref{eq:numerical_scheme} for $M=N$, $M=100$ and $M=1000$ and $u$ solving \eqref{eq:1DL_system}. These results correspond to the same setting as in Figure \ref{fig:error_plot}.}\label{fig:runtime}
\end{figure}
Figure \ref{fig:error_plot} shows the $\ell^2$-norm of the error between the effective field $\psi$ and the true solution $u$ at each resonator $z_i$ as a function of $N$. It becomes apparent from our numerical result that the error behaves algebraically in $N$. This result is not surprising,  as the effective medium theory only holds in the limit $\ell\to0$ and $N\to\infty$. In Figure \ref{fig:error_plot}, we specifically compare the resulting $\ell^2$-error for different choices of $M$. We see that the error between the effective field and the true solution decreases with increasing $N$ and $M \leq N$. Then, it is sufficient to choose $M=N$ for $N$ large enough.\par 
The results in Figure \ref{fig:runtime} show how the computation run time of $u$ is significantly larger than the computation of the effective field $\boldsymbol{\psi}$ for $M=N$, $M=100$ and $M=1000$. This further underscores the relevance of the effective medium theory derived herein. Note that the run time of the computation of $\boldsymbol{\psi}^M$ for fixed $M (=100, 1000)$ is (almost) constant in $N$.

\section{Concluding Remarks}\label{sec:conclusion}
We have rigorously derived an effective medium theory for one-dimensional time-modulated metamaterials in the low-frequency regime. We started by providing the $2\times2$ scattering matrix for static metamaterials and this led to the point interaction approximation. This set the ground for following a similar approach as in \cite{Hai_Habib}. With suitable assumptions \ref{ass:1} - \ref{ass:riemannsum_limit} we derived the homogenised equation \eqref{eq:homogeq_static}.\par 
We then proceeded with the case of the time-dependent material parameter $\kappa$. In this paper, we assumed the parameter $\rho$ to be static, since in \cite{jinghao_liora} we proved that the resulting wave field does not depend on $\rho$ at leading order. As in the static case, we first derived the scattering matrix. However, for a time-modulated material parameter, this is a $2(2K+1)\times2(2K+1)$ matrix. This is a direct consequence of the mode coupling that arises from the modulation in time. Analogously to the static case we then derived a point interaction approximation, which ultimately led to a characterisation of the effective field through a system of coupled Lippmann-Schwinger equations \eqref{eq:system_LippmannSchwinger}. These integral equations furnished the homogenised equations given by \eqref{eq:homog_eq_tdep}. In contrast to the static case, the homogenised equations modelling a time-modulated metamaterial are a system of coupled differential equations.\par
A crucial assumption to obtain a time-modulated effective medium theory is that the incident frequency $\omega$ is slightly away from the resonant frequency of the components. For a frequency at resonance, the scattering coefficients $g$ are of order one. Even for a large number $N$ of resonators, each individual resonator makes a strong contribution to the total field, and the limit does not exist. In summary, we showed that time-dependent metamaterials at a resonant frequency cannot be treated as an effective medium.\par
Finally, we introduced a numerical scheme to solve the Lippmann-Schwinger equation for the effective field. Our numerical solution supports our analytical results by showing that the numerical solution of \eqref{eq:numerical_scheme} tends to the numerical solution of \eqref{eq:1DL_system}. In fact, Figure \ref{fig:error_plot} shows that the error is algebraic in $N$.\par
We consider the results proven in this paper to be the basis for many new breakthroughs in the exploration of time-modulated metamaterials, similarly to achievements regarding static metamaterials \cite{Hai_Habib}. Since we expect our results to generalise to higher dimensions (see Remark \ref{rem:higherdim}), the results herein proved are also of great significance to two- and three-dimensional spacetime metamaterials. 

\section*{Code availability}
The codes that were used to generate the results presented in this paper are available under \url{https://github.com/rueffl/effective_medium_theory_timedep}.

\addcontentsline{toc}{chapter}{References}
\renewcommand{\bibname}{References}
\bibliography{refs}
\bibliographystyle{plain}

\end{document}